\documentclass{siamltex1213}

\usepackage{amssymb}
\usepackage[sumlimits]{amsmath}

\newtheorem{thm}{Theorem}[section]
\newtheorem{lem}[thm]{Lemma}
\newtheorem{prop}[thm]{Proposition}
\newtheorem{cor}[thm]{Corollary}
\newtheorem{rmk}{Remark}

\newcommand{\ind}{\mbox{$\perp \kern-5.5pt \perp$}}

\newcommand{\mast}{\mathrm{MAST}}

\newcommand{\blue}[1]{#1}

\title{Bounds on the Expected Size of the Maximum Agreement Subtree} 
\author{
Daniel Irving Bernstein\thanks{Department of Mathematics, North Carolina State University, Raleigh, NC, USA; 
{\tt\{dibernst,celong2,smsulli2\}@ncsu.edu}.}
\and
Lam Si Tung Ho\thanks{\blue{Department of Biostatistics}, University of California, Los Angeles, CA, USA;   
{\tt lamho@ucla.edu}}
\and
Colby Long\footnotemark[1]
\and
Mike Steel\thanks{Biomathematics Research
Centre, University of Canterbury, Christchurch, New Zealand
{\tt mike.steel@canterbury.ac.nz}}
\and 
Katherine St.~John\thanks{Department of Mathematics and Computer Science, Lehman College, City University of New York, Bronx, NY, USA; {\tt katherine.stjohn@lehman.cuny.edu}; 
Department of Invertebrate Zoology, American Museum of Natural History, New York, NY, USA.}
\and
Seth Sullivant\footnotemark[1]
}

\begin{document}
\maketitle



\begin{abstract}
We \blue{prove} lower bounds on the expected
size of the maximum agreement subtree of two random binary phylogenetic 
trees under both the uniform distribution and Yule-Harding
distribution \blue{and prove upper bounds under the Yule-Harding distribution}.  This positively answers a question
posed in earlier work.  Determining tight upper and lower bounds
remains an open problem.
\end{abstract}


\section{Introduction}

Leaf-labelled trees are a canonical model for evolutionary histories of sets of species \cite{sem}.
Let $T_{1}$ and $T_{2}$ be two trees with the
same set of leaf labels $X$ (interior vertices are unlabelled).  
\blue{Following \cite{sem}, a \emph{rooted} tree is a tree that has exactly one distinguished vertex called the \emph{root}.}
A subset $S \subseteq X$
yields an {\em agreement subtree} of $T_{1}$ and $T_{2}$
if $T_{1}|_{S}  = T_{2}|_{S}$ where for a tree
$T$, $T|_{S}$ is the tree restricted to the
leaf label set $S$ and is obtained by supressing
all vertices of degree $2$ \blue{(excepting the root, if $T_1$ and $T_2$ are rooted)}.  A {\em maximum agreement
subtree} is a subtree that is an agreement subtree
of the maximal size in $T_{1}$ and $T_{2}$ (see Figure~\ref{fig:MAST}). We note that
there might be multiple  maximum agreement subtrees for
the pair $T_{1}$ and $T_{2}$.  Let $\mast(T_{1}, T_{2})$
denote the number of leaves of 
a maximum agreement subtree of $T_{1}$ and $T_{2}$, which can be computed in polynomial time in $|X|$ \cite{war93}.

Let $T_{1}$ and $T_{2}$ be two unrooted binary trees with $n$ leaves. It is known that $\mast(T_{1}, T_{2}) = \Omega( \sqrt{\log n})$
for any pair of trees \cite{Martin2013},
and there is always a pair of trees $T_{1}$ and $T_{2}$ such that
$\mast(T_{1}, T_{2}) = O(\log n)$.  Closing the gap
on this worst case behaviour is a lingering open problem.
This worst case behaviour is quite different if the two input trees are
rooted (for rooted trees any agreement subtree is also required
to respect the induced rooting); in this case, it is easily seen
that, for any $n\geq 2$ there is always a pair of trees $T_{1}$ and $T_{2}$ such that $ \mast(T_{1}, T_{2}) =2$.

Of practical interest is to understand what the expected size
of the maximum agreement subtree when
$T_{1}$ and $T_{2}$ are drawn from a suitable distribution
on the set of all binary trees.
For example, de Vienne, Giraud, and Martin \cite{deVienne2007} proposed
using the maximum agreement subtree as a measure of the
congruence between two trees.  Understanding the distribution
of this statistic can be used in hypothesis tests of the null
hypothesis that the two trees were generated at random \cite{lap}. 
For example, the deviation from the null hypothesis between
a host tree and a parasite tree could be used as evidence
of co-speciation \cite{hous}.

The mathematical study of the distribution of
the size of the maximum agreement subtree was initiated in the work 
of Bryant, McKenzie, and Steel \cite{Bryant2003}.
They specifically focused on the expectation
$f_{u}(n)  =  \mathbb{E}[\mast(T_{1}, T_{2})]$ and
$f_{YH}(n) = \mathbb{E}[\mast(T_{1}, T_{2})]$
where in first case the trees are drawn independently from the uniform distribution
on \blue{rooted} binary phylogenetic trees with $n$ leaves, 
and the second case, from the Yule-Harding distribution on \blue{rooted} binary phylogenetic trees with $n$ leaves \blue{(for a formal definition of this distribution, see \cite{sem}, Section 2.5)}.
Their simulations, for trees with up to $n = 1024$ leaves, suggest
that under both the uniform distribution on binary trees
and the Yule-Harding distribution \cite{harding1971},
the expected size of the maximum agreement subtree is
of order $\Theta( n^{a})$ with $a \approx 1/2$,
and they also proved that $f_{u}(n)  =  O(n^{1/2})$.

In the special case when both random trees are
caterpillar trees, finding the maximum agreement
subtree is essentially equivalent to finding the
longest increasing subsequence in a random permutation.
This problem has a long history, and it is well-known
that the expected size of the longest increasing subsequence is
asymptotically $2 \sqrt{n}$, and that the (appropriately rescaled) distribution
of the longest increasing subsequence is the Tracy-Widom distribution 
(see \cite{Aldous1999} for a survey of results).
The distribution of the maximum agreement
subtree is a natural extension of the longest increasing
subsequence problem to trees.

Bryant, McKenzie, and Steel \cite{Bryant2003} posed the question, and again suggested the problem
at the 2007 Newton Institute program on Phylogenetics,
of finding any exponent $a>0$ such that
$f_{u}(n)  =  \Omega( n^{a})$ or $f_{YH}(n) = \Omega(n^{a})$.
The main results of this note are, \blue{for rooted binary trees, to derive the conjectured (power law type)} lower bounds for the expected size of the maximum agreement subtree 
for both the uniform and Yule-Harding distributions,
 and an upper bound of the form $O(n^{1/2})$
for the Yule-Harding distribution.

\blue{ Note that the uniform and Yule-Harding distribution satisfy two fundamental properties, namely {\em exchangeability} and {\em sampling consistency}. Exchangeability means that if two trees $T_1$ and $T_2$ 
differ only by a permutation of the leaves, then $\mathbb{P}_s(T_1) =\mathbb{P}_s(T_2)$.  
Sampling consistency is the condition that for any subset $S$ of $X$ if we generate a rooted binary tree $T$ on leaf set $X$  under either model (uniform or Yule-Harding) then the induced tree $T|_S$
is described by the same model (uniform or Yule-Harding, respectively).
For further details see \cite{Bryant2003}.
}

\begin{figure}
\centering
\includegraphics[height=2in, width=4in]{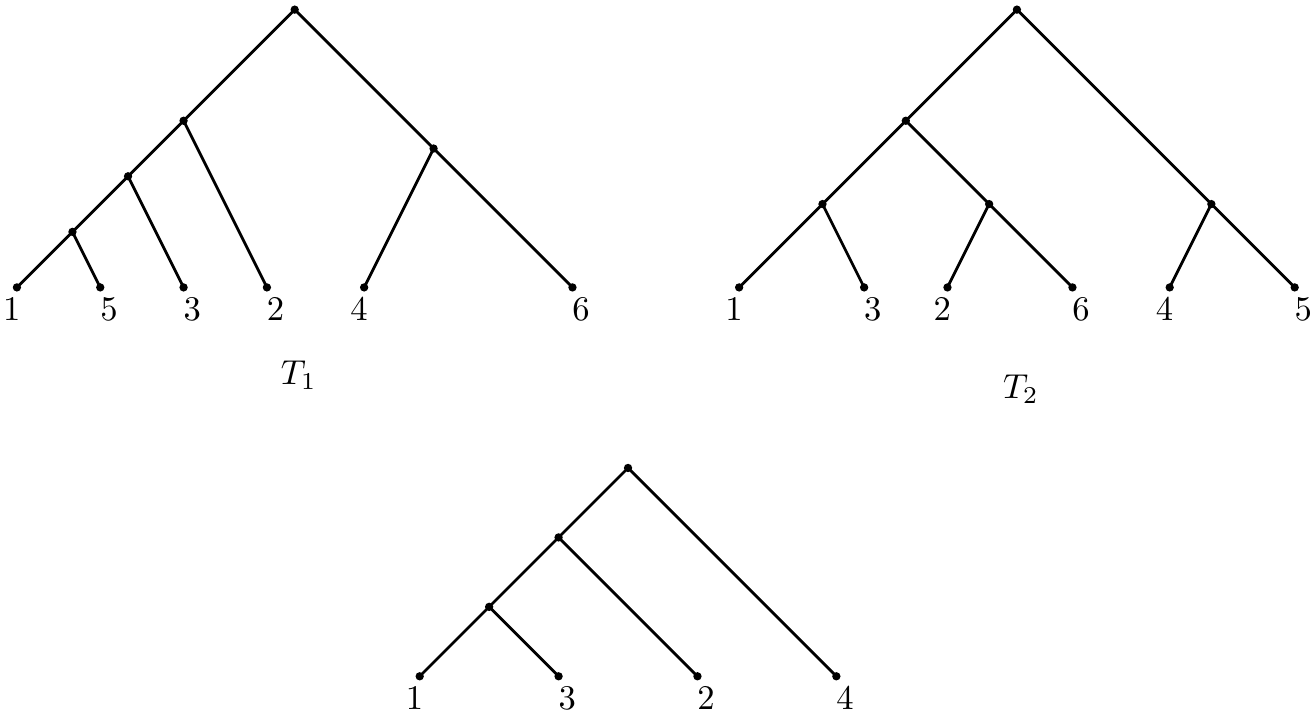}
\caption{Two rooted phylogenetic trees, $T_1$ and $T_2$, on 6 leaves and a maximum agreement subtree (MAST) for $T_1$ and $T_2$.  The MAST illustrated is a {\em caterpillar} with leaves encoding the subsequence, $(1,3,2,4)$. }
\label{fig:MAST}
\end{figure}


\section{Uniform Trees}

To show our lower bound results for \blue{rooted} binary trees chosen from a uniform distribution, we rely on classical results on the expected largest increasing subsequence in a \blue{random} permutation of numbers.
For trees of size $n$ \blue{(both unrooted and rooted)} we show that the expected length of a caterpillar subtree is $\Omega(\sqrt{n})$.
We can then show that \blue{for two trees $T_1, T_2$ chosen independently and uniformly at random from $RB(n)$} there is a subset, $S' \subseteq [n]$ of $\Omega(n^{1/4})$ leaves which induce rooted caterpillars of $T_1$ and $T_2$.  Restricting to this subset $S'$, we can view $T_{1}|_{S'}$ and $T_{2}|_{S'}$ as permutations of the elements of $S'$ and apply the classical results of Aldous and Diaconis \cite{Aldous1999} to yield a common subsequence of length ${|S'|}^{1/2} = \Omega((n^{1/4})^{1/2}) = \Omega(n^{1/8})$.

Let $RB(n)$ denote the set of rooted binary phylogenetic trees
with leaf label set $[n] := \{1,2,3, \ldots, n\}$.
Similarly, $B(n)$ denotes the set of unrooted binary phylogenetic
trees with leaf label set $[n]$.
Note that $|RB(n)|  = (2n -3)!! = 1 \times 3 \times 5  \times  \cdots \times (2n -3) $, and $b(n) = |B(n)| = (2n-5)!!$.
In this section, we consider the uniform distribution on 
$RB(n)$, and the function $f_u(n)  =  \mathbb{E}[ \mast(T_1, T_2)]$
where $T_1$ and $T_2$ are generated uniformly and independently from $RB(n)$.

\begin{thm}\label{Ulowerbound}
    For any \blue{$\lambda < 2^{1/4}(1 - \frac{e^{-1/4}}{2})^2 \approx .443$} there is a value $m$ so that, for all $n \geq m$,
    $$
    f_u(n)  \geq  \lambda  n^{1/8}.
    $$
\end{thm} 

Let $T$ be an unrooted binary tree with leaves labeled by $[n]$.
Then for $i,j \in [n]$, $d_n(i,j)$ denotes the number of edges on the unique path from leaf $i$ to leaf $j$.

\begin{prop}\label{exact}
	\blue{Let $T$ selected uniformly at random from $B(n)$.
	For leaves $i,j \in [n], i \neq j$, the probability that $d_n(i,j) = m$
	is:}
	\[
		\mathbb{P}(d_n(i,j)=m)=\frac{(n-2)!}{(2n-4)!}\cdot\frac{2^{m-1}(m-1)(2n-m-4)!}{(n-m-1)!}.
	\]
\end{prop}
\begin{proof}
	Let $D(m,n)$ denote the number of $[n]$-trees where $d_n(i,j)=m$.
	If $T$ is an unrooted $[n]$-tree with $d_n(i,j) = m$,
	then we can draw $T$ as follows:
	
	\begin{center}
        	\includegraphics{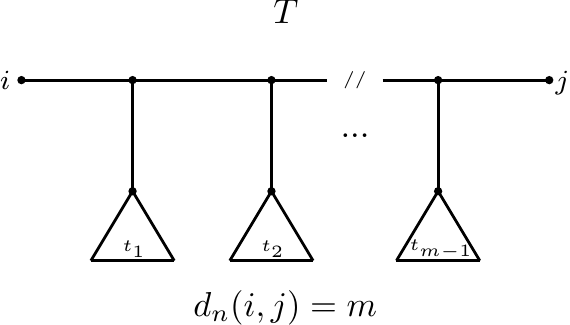}
	\end{center}

	\noindent thus giving us a bijection between the set of trees where $d_n(i,j) = m$,
	and the set of ordered forests \blue{consisting of rooted binary} trees on $n-2$ leaves (as noted by \cite{Steel1993} for a different calculation).
	Now, the set of ordered forests on $m-1$ rooted binary trees and $n-2$ leaves is just $(m-1)! N(n-2, m-1)$
	where $N(r, k)$ is the number of unordered forests of $k$ \blue{rooted binary} trees and $r$ leaves.  Now,
	$N(r, k) = \frac{(2r-k-1)!}{(r-k)!(k-1)!2^{r-k}}$ for $r \geq k$ and $0$ if $r<k$, as stated in Lemma 4 of \cite{car}. This result can be derived by observing that:
	\begin{equation}
	\label{nr}
		N(r,k) = r!\cdot\left[x^{r}\right]B(x)^{k},
\end{equation}
	where $B(x)= 1-\sqrt{1-2x}$ is the  exponential generating function for the number of \blue{rooted} binary trees on $k$ non-root leaves, and applying the Lagrange inversion
	formula, together with the identity $B(x) = x(1-\frac{1}{2}B(x))^{-1}$, to determine the RHS of (\ref{nr}) (for further details, see \cite{sem}, Section 2.8). 
	This gives the following expression for $D(m,n)$:
	\[
		D(m,n) = \frac{(m-1)\cdot(2n-m-4)!}{2^{n-m-1}\cdot(n-m-1)!}.
	\]
	Dividing the above by $b(n) = (2n-5)!!=\frac{(2n-4)!}{(n-2)!\cdot 2^{n-2}}$ gives the desired formula for $P(d_n(i,j)=m)$.
\end{proof}

Using Proposition~\ref{exact}, we calculate the probability that the path length between two leaves, $i$ and $j$ exceeds $\sqrt{n}$:
\begin{lem} \label{bound}
\label{c}
$\displaystyle \lim_{n \to \infty} \mathbb{P}(d_n(i,j) \geq \sqrt{n}) \geq \left (1 - \dfrac{e^{-1/4}}{2} \right ) = c\approx 0.61.$
\end{lem}

\begin{proof}
	For a fixed $n$, 
	\[
		\mathbb{P}(d_n(i,j) = m + 1) - \mathbb{P}(d_n(i,j) = m ) = \frac{2^{m - 1}(2n - m -5)!(n - 2)!}{ (n - m - 1)!(2n - 4)!}(-m^2 + m + 2n - 4),
	\]
	which is positive whenever $m \leq \sqrt{2n}$.
	Therefore, we have, 
	\begin{align*}
		\mathbb{P}(d_n(i,j) < \sqrt{n}) &= \displaystyle \sum_{m = 1}^{\blue{\lceil \sqrt{n} - 1 \rceil}} \mathbb{P}(d_n(i,j) = m ) \\
		& \leq \sqrt{n}(\mathbb{P}(d_n(i,j) = \sqrt{n}). \\
	\end{align*}
	Using Stirling's approximation for all factorials we have:
	\begin{align*}
	\sqrt{n}(\mathbb{P}(d_n(i,j) = \sqrt{n}) &\sim \frac{e}{2} \frac{n 2^{\sqrt{n}} (n - 2)^{n - 2} (2n - \sqrt{n} - 4)^{2n - \sqrt{n} - 4}}{(2n - 4)^{2n - 4}(n - \sqrt{n} -1)^{n - \sqrt{n} - 1}} \\
	& \sim \frac{e}{2} \left ( \frac{2n - \sqrt{n} - 4}{2n - 4} \right )^{2n} \left ( \frac{n-2}{n - \sqrt{n} - 1} \right )^{n}  \left (\frac{2n - 2\sqrt{n} - 2}{2n - \sqrt{n} - 4} \right )^{\sqrt{n}} \\
	& \sim \frac{e}{2} \left( 1 - \frac{3n - 8}{4(n-2)(n - \sqrt{n} - 1)} \right)^n \left ( 1 - \frac{\sqrt{n} - 2}{2n - \sqrt{n} - 4} \right )^{\sqrt{n}}\\
	& \sim \frac{e}{2} e^{-3/4} e^{-1/2} = \frac{e^{-1/4}}{2}.
	\end{align*}
	Hence, $\displaystyle \lim_{n \to \infty}  \mathbb{P}(d_n(i,j) < \sqrt{n}) \leq \dfrac{e^{-1/4}}{2}.$
	
	Since $\mathbb{P}(d_n(i,j) \geq \sqrt{n}) = 1 -  P(d_n(i,j) < \sqrt{n})$, we have  
	\[
		\displaystyle \lim_{n \to \infty} \mathbb{P}(d_n(i,j) \geq \sqrt{n}) \geq \left (1 - \dfrac{e^{-1/4}}{2} \right ).
	\]
\end{proof}

Recall that a \emph{rooted caterpillar} on $n$ leaves is any rooted binary phylogenetic tree for which the induced subtree on the interior vertices forms a path graph with the root at one end of the path. We show that two trees, chosen uniformly and independently from rooted binary trees on $n$ leaves, have a common rooted caterpillar of height at least $\Omega(n^{1/4})$ leaves:

\begin{lem}
\label{caterpillar}
 Let $T_1, T_2$ be rooted $n$-leaf trees chosen uniformly and independently from $RB(n)$. If $n$ is sufficiently large, \blue{then} with probability greater than $(1 - \frac{e^{-1/4}}{2})^2 = c^{2}$, there exists $S' \subset [n]$ with $|S'| \geq \frac{1}{2\sqrt{2}}n^{1/4}$ such that both $T_{1}|_{S'}$ and $T_{2}|_{S'}$ are rooted caterpillar trees.
\end{lem}

\begin{proof}
 Let $T_1$ and $T_2$ be chosen uniformly and independently from $RB(n)$. Choose leaves $i$ and $j$ uniformly at random and temporarily regard $T_1$ as an unrooted tree by suppressing the root vertex. If $d_{ij}$ is the distance between $i$ and $j$ in the unrooted tree, then, in the rooted tree, the distance from the root vertex to either  $i$ or $j$ must be greater than or equal to $\frac{1}{2}d_{ij}$. 
Therefore, the probability that $T_1$ has height at least $m$ is greater than or equal to the probability that $d_{ij} \geq 2m$.
If $T_1$ has height $k$ then we can choose $S \subset [n]$ with $|S| = k$ so that $T_{1}|_{S}$ is a rooted caterpillar. By Lemma \ref{c}, the probability of finding $S\subset [n]$ with $|S| \geq \frac{1}{2}n^{1/2}$ such that $T_{1}|_S$ is a rooted caterpillar is greater than $c$. If such an $S$ exists, then since $T_2$ was chosen uniformly from $RB(n)$ and independently from $T_1$, $T_{2}|_S$ is a tree chosen uniformly from $RB(|S|)$. Applying Lemma \ref{c} again, the probability that there exists an $S'\subset [|S|]$ with $|S'| \geq \frac{1}{2}\sqrt{|S|}$  such that  $T_{2}|_{S'}$ is a rooted caterpillar tree is also greater than $c$. Since the restriction of a rooted caterpillar is a rooted caterpillar the result follows.
\end{proof}

Now let $(T_1,T_2)$ be a pair of trees satisfying the conditions of Lemma \ref{caterpillar}. Select the set $S'$ with  \blue{ $|S'| = {q(n)} := \lfloor \frac{1}{2\sqrt{2}}n^{1/4} \rfloor $}
and relabel the leaves of both $T_1$ and $T_2$ so that when drawn with the leaf vertex adjacent to the root on the left, the leaf labels of $T_{1}|_{S'}$ increase from left to right. Draw $T_2$ in the same way picking either representation for the leaves of the cherry in $T_2$ with equal probability. 
The order of the leaves of $T_2$ gives a permutation uniformly chosen from the set of permutations of $[q]$.  From Aldous and Diaconis \cite{Aldous1999}, we have:
 
\begin{thm}
\label{subsequence}
  \cite[Theorem 2]{Aldous1999} Let $\pi_n$ be a uniform random permutation of $[n]$\blue{. Define} the integer valued random variable $L_n := l(\pi_n)$ where $l(\pi)$ is the length of the longest increasing subsequence of $\pi$. Then $\mathbb{E}[L_n] {\raise.17ex\hbox{$\scriptstyle\sim$}} 2n^{1/2}$ as $n \rightarrow \infty$.
\end{thm}

Observe that if $\sigma$ is the permutation of $\blue{[q(n)]}$ given by $T_{2}|_{S'}$ and $s(\sigma)$ is
the set of elements of \blue{an increasing subsequence} of $\sigma,$ then $T_{1}|_{s(\sigma)} = T_{2}|_{s(\sigma)}$ \blue{implying} $M(T_1, T_2) \geq l(\sigma)$. 
In other words, $ \mathbb{P} (\mast(T_1, T_2) = i ) $ is at least the product of the probability that $T_1$ and $T_2$ restrict
to a rooted caterpillar of size $\blue{q(n)}$ and $ \mathbb{P}(L_{\blue{q(n)}} = i) $.
This is the key observation that will allow us to prove Theorem \ref{Ulowerbound}.

\begin{proof}[of Theorem \ref{Ulowerbound}]:
By Lemma \ref{c} and Theorem \ref{subsequence} and the observations above, \blue{as $n \rightarrow \infty$}
\begin{align*}
     f_u(n) 
    &\blue{>}  \displaystyle \sum_{i = 1}^{\blue{q(n)}} i\mathbb{P}(M(T_{1}, T_{2}) =i) \\
    & >  \displaystyle \sum_{i = 1}^{\blue{q(n)} } i(c^2 \mathbb{P}(L_{\blue{q(n)} } = i)) \\
    &= c^2  \mathbb{E}[L_{\blue{q(n)}}]. \\
    &\sim \blue{2^{1/4}c^2n^{1/8}}.  
\end{align*}
\end{proof}



\section{Yule-Harding Trees:  Lower Bounds}

In this section we derive our lower bounds on the expected size
of the maximum agreement subtree under the Yule-Harding distribution \cite{harding1971}.
The Yule-Harding distribution is a probability distribution on \blue{the set of rooted binary}
trees that is defined in a constructive manner, by building
up a tree on $n$ leaves by successively adjoining leaves.
A Yule-Harding tree on $n$ leaves is obtained from a Yule-Harding
tree on $n-1$ leaves by  choosing a leaf uniformly at  randomly
and branching that leaf into two new leaves.  Leaf labels for the \blue{$(n-1)$-leaf} tree are chosen as \blue{a uniformly random subset} of size $n-1$ from
$[n]$.
Let $f_{YH}(n)$ denote the expected size of the maximum-agreement
subtree between two trees from $RB(n)$ sampled independently from the Yule-Harding distribution.

\begin{thm}\label{thm:extreme}
Let $a$ be the unique positive root of the equation
$
2^{2 - a}  =  (a + 1)(a+2)
$
(approximately $a = .344184...$).  Then for any $\epsilon > 0$, 
$f_{YH}(n) = \Omega(n^{a - \epsilon})$.
\end{thm}
We abbreviate $f_{YH}(n) = f(n)$ in the arguments below.
\blue{We first establish a number of preliminary results that are needed for the proof of Theorem ~\ref{thm:extreme}.}

\begin{lem}\label{lem:discreteintegral}
Let $b \geq 0$.  Then
\blue{$$
 \frac{1}{b+1}k^{b+1} \leq \sum_{i =1}^{k}  i^{b}  \leq \frac{1}{b+1}(k+1)^{b+1}.
$$}
\end{lem}

\begin{proof}
This follows by applying the left and right end-point rules for the integral
of  $\int_{0}^{k}  x^{b} dx$.
\end{proof}

\blue{Next we } calculate lower bounds on the overlap of any two splits of $[n]$ (bipartitions of the leaves):

\begin{lem}\label{lem:sizereq}
Let $A_{1}|B_{1}$ and $A_{2}|B_{2}$ be two splits of $[n]$ with
$|A_{1}| = i$,  $|A_{2}| = j$ and $i \leq j \leq n/2$.  
Then either
\begin{eqnarray*}
|A_{1} \cap A_{2}|  \geq \lceil i/2 \rceil  \mbox{ and } |B_{1} \cap B_{2}| 
\geq j - \lfloor i/2 \rfloor  &   & \mbox{ or }  \\
|A_{1} \cap B_{2}|  \geq \lceil i/2 \rceil  \mbox{ and } |B_{1} \cap A_{2}| 
\geq j - \lfloor i/2 \rfloor.  &  & 
\end{eqnarray*}
\end{lem}

\begin{proof}
Make a $2 \times 2$ matrices whose entries are the four intersection values:
$$ M = 
\begin{pmatrix}
|A_{1} \cap A_{2}|  & |A_{1} \cap B_{2}|  \\
|B_{1} \cap A_{2}|  &  |B_{1} \cap B_{2}| 
\end{pmatrix}.
$$
The row sums of $M$ are $i, n-i$ and the column sums are $j, n-j$.
So either $M_{11}$ or $M_{12}$ are $\geq \lceil i/2 \rceil $.
If $M_{11} \geq \lceil i/2 \rceil$ then $M_{22} \geq n- j - \lfloor i/2 \rfloor \geq
 j - \lfloor i/2 \rfloor$, since $j \leq n/2$.
If $M_{12} \geq \lceil i/2 \rceil$ then $M_{21}  \geq j - \lfloor i/2 \rfloor$.
\end{proof}

We can use Lemma \ref{lem:sizereq} in a worst case
analysis to get lower \blue{bounds on  $f(n)$}.

\begin{lem}
Let $n = 2k+1$ be odd.  Then
$$
f(2k+1)  \geq  \frac{8}{(n-1)^{2}}  
\sum_{1 \leq i < j \leq k}  (f( \lceil i/2 \rceil) + f(j - \lfloor i/2 \rfloor) )
+  \frac{8}{(n-1)^{2}}    \sum_{1 \leq i \leq k}  f( \lceil i/2 \rceil).
$$
Let $n = 2k$ be even.  Then

\begin{eqnarray*}
f(2k) &  \geq &   \frac{8}{(n-1)^{2}}  
\sum_{1 \leq i < j < k}  (f( \lceil i/2 \rceil) + f(j - \lfloor i/2 \rfloor) )
+ \frac{8}{(n-1)^{2}}    \sum_{1 \leq i < k}  f( \lceil i/2 \rceil)  \\
 &  &  + \frac{4}{(n-1)^{2}}  
 \sum_{1 \leq i < k}  (f( \lceil i/2 \rceil) + f(k - \lfloor i/2 \rfloor) )
 + \frac{2}{(n-1)^{2}}    f( \lceil k/2 \rceil)
\end{eqnarray*}
\end{lem}

\begin{proof}
For two discrete random variables $X$ and $Y$, the law of total expectation
says
\begin{equation}
\mathbb{E}[X]  =  \sum_{y} P(Y =y)  \mathbb{E}[X | Y = y].
\label{eqn:totalE}
\end{equation}

We use this identity to get lower bounds \blue{on} $f(n)$.
In particular, we condition on the event that the daughter
subtrees of the root in $T_{1}$ and $T_{2}$ have sizes $(i, n - i)$  
and $(j, n-j)$ respectively, and so we apply (\ref{eqn:totalE}) with $Y=(i,j)$ and $X= \mast(T_1, T_2)$.
Since we are sampling from the Yule-Harding model, 
the size of the daughter subtrees of $T_1$ (and of $T_2$) follows a uniform distribution (\blue{see \cite{slowinski1990} and also} follows directly from basic P\'{o}lya Urn theory \cite{mahmoud}). Thus, the probability of the conditioning event is
$\frac{1}{(n-1)^{2}}$ for any $i, j \in [n-1]$. 

By the \blue{symmetry of the problem, it suffices to look}
at pairs $1 \leq i \leq j \leq \lfloor n/2 \rfloor$.  This is where
the factor of $8$ comes from in the various expressions.
Once we fix $i$ and $j$, we are restricting to the case where two random Yule-Harding trees have daughter subtrees of the root
$A_{1}|B_{1}$ and $A_{2}|B_{2}$ where $|A_{1}| = i$ and $|A_{2}| = j$.
 To get a lower bound on the conditional expectations, we can
reduce to the two subtrees on either $A_{1} \cap A_{2}$ and $B_{1} \cap B_{2}$
or  $A_{1} \cap B_{2}$ and $B_{1} \cap A_{2}$,
depending on which one satisfies the size requirements from Lemma\ref{lem:sizereq}.  On those two induced subtrees we have an expected
$\mast$ of size at least $f( \lceil i/2 \rceil)$  and $f(j - \lfloor i/2 \rfloor)$
respectively.  Combining those trees together through the common root
gives a $\mast$ of expected size $f( \lceil i/2 \rceil) + f(j - \lfloor i/2 \rfloor)$.
The formulas follow by analysis by cases with attention to double counting and the boundary cases of $i = j$, $i < j = n/2$ and $i = j = n/2$.
\end{proof}

As a direct application, we have by counting the number of occurrences of $f(i)$:

\begin{cor}\label{cor:simple}
$$
f(n)  \geq  \frac{8}{(n-1)^{2}}  \sum_{i = 1}^{\lfloor n/2 \rfloor}
(n - 2i) f(i).
$$
\end{cor}


The lower bound for the Yule-Harding case follows by induction on the inequality of Corollary~\ref{cor:simple}:

\begin{proof}[of Theorem~\ref{thm:extreme}]
We use the inequality from Corollary \ref{cor:simple},
together with induction.  Clearly $f(1) =1  \geq  c \times 1^{a}$ for $c =1$.  
Assume that $f(k) \geq  c k^{a}$ for all $k < n$, then we have:
$$
f(n)  \geq  \frac{8}{(n-1)^{2}}  \sum_{i = 1}^{\lfloor n/2 \rfloor} 
(n - 2i) \times c i^{a}. 
$$
Applying  Lemma~\ref{lem:discreteintegral} we deduce
\begin{eqnarray*}
f(n) & \geq & \frac{8}{(n-1)^{2}} \cdot c \cdot  \left(  \frac{n}{a + 1} {\lfloor \frac{n}{2} \rfloor}^{a+1} 
- \frac{2}{a + 2} {\lfloor \frac{n}{2}  + 1\rfloor}^{a+2}   \right)  \\
&  = &  \frac{2^{2 - a}}{(a + 1)(a+2)}  \cdot c \cdot 
\frac{(a + 2) n (2 \lfloor \frac{n}{2} \rfloor)^{a+1}  -  (a+1) ( 2  \lfloor \frac{n}{2}  + 1\rfloor) ^{a+2}  }{(n-1)^{2}}.  
\end{eqnarray*}

Note that the rightmost expression in the product is asymptotic to
$n^{a}$, converging to it from below.
In particular, for any $\delta > 0$, there exists \blue{an} $N$ such that for
all $n > N$, we have  
$$
\frac{(a + 2) n (2 \lfloor \frac{n}{2} \rfloor)^{a+1}  -  (a+1) ( 2  \lfloor \frac{n}{2}  + 1\rfloor) ^{a+2}  }{(n-1)^{2}}  > (1 - \delta) n^{a}.
$$
This yields 
\begin{eqnarray*}
f(n) & \geq  &  \frac{2^{2 - a}}{(a + 1)(a+2)} (1 - \delta)  \times c n^{a}.
\end{eqnarray*}
To complete the induction we must have 
$$\frac{2^{2 - a}}{(a + 1)(a+2)} (1 - \delta)  \geq 1.$$
Since we can take $\delta$ arbitrarily small, this completes the result.
\end{proof}

Note that in this argument the value of $c$ will depend on $\delta$ (through the interaction with $N$).  Hence, we cannot use the proof argument to take $\epsilon = 0$ in the statement.

\begin{rmk}
Further modifications to the above argument can be made to get 
slight improvements
on the exponent.
For example, when $i$ is very small, instead of taking the subsets
of size $\geq \lceil i/2 \rceil$ and $\geq j - \lfloor i/2 \rfloor$, 
passing to a single subset of size $n-i$ (throwing out the subset of size $i$)
can yield an improvement in the bounds.  That is, if $i$ is small then
$$
f(\lceil i/2 \rceil) + f( j - \lfloor i/2 \rfloor)  \leq  f(n-i)
$$
when $f(n) = n^{\alpha}$ and $\alpha$ bounded away from $0$.  
Using this reasoning coupled with the arguments above, we were able to 
increase the exponent in the theorem to approximately $.384$.
\end{rmk}


\section{Yule-Harding Trees:  Upper Bounds}

In this section, we derive $O(n^{1/2})$ upper bounds on 
the expected size of the maximum agreement subtree for any
distribution on trees that is exchangeable and
satisfies sampling consistency \blue{(described at the end of the Introduction)} based on ideas from
\cite{Bryant2003}.  
We just need 
the crucial Lemma 4.1 from that paper.
For a fixed distribution on trees let $\mathbb{P}_s(t)$
be the probability  of the tree
$t$ which has $s$ leaves.

\begin{lem}\cite[Lemma 4.1]{Bryant2003}
Suppose that phylogenetic trees $T_1$ and $T_2$ on a leaf set $L$ of size $n$ are randomly
generated under a model that satisfies exchangeability and sampling 
consistency.  Then
$$
\mathbb{P}[\mast(T_1, T_2) \geq s ]  \leq  \psi_{n,s} =  { n \choose s}  \sum_{t \in RB(s)}
\mathbb{P}_s(t)^2.
$$ 
\end{lem}

From here, we choose a function $s = g(n)$  so that 
${ n \choose s}  \sum_{t \in RB(s)} \mathbb{P}_s(t)^2$ 
tends rapidly to zero with $n$, to deduce that  $\mathbb{E}[\mast(T_1, T_2)] = 
O(g(n))$.

\begin{prop}
Let $\mathbb{P}_s$ be any exchangeable distribution on
rooted binary trees.  Then
$$
\sum_{t \in RB(s)}
\mathbb{P}_s(t)^2  \leq   \frac{2^{s-1}}{s!}.
$$
\end{prop}

\begin{proof}
Let $\mathbb{P}_s$ be any exchangeable distribution on
rooted binary trees. 
Note that if $0 < x \leq y$ then
$$
x^2 + y^2  \leq (x- \epsilon)^2 + (y + \epsilon)^2 = x^2 + y^2 + 2 \epsilon^2 + 2\epsilon(y-x). 
$$
This implies that if we take $\mathbb{Q}_s(t)$ to \blue{be} the probability
distribution that puts zero mass on trees that have a shape different from a tree $t'$ that maximizes  $\mathbb{P}_s(t')$, we will have
$$
\sum_{t \in RB(s)}
\mathbb{P}_s(t)^2  \leq  \sum_{t \in RB(s)} \mathbb{Q}_s(t)^2.
$$
By exchangeability, when $\mathbb{Q}_s(t) \neq 0$, it is $\frac{1}{NT(t)}$
where $NT(t)$ is the number of rooted trees with tree shape $t$. Thus,
$$
\sum_{t \in RB(s)} \mathbb{Q}_s(t)^2  =  \blue{\frac{1}{NT(t)^2}+\cdots + \frac{1}{NT(t)^2} \mbox{ ($NT(t)$ times)} = \frac{1}{NT(t)}}.
$$
To maximize this quantity, we choose a tree shape with the fewest number of
trees with that tree shape.  The number of trees of a given shape $t$ is
$s!/2^{m}$ where $m$ is the number of internal vertices of $t$ that are symmetry vertices \blue{(i.e. vertices of $t$ for which the two daughter subtrees have the same shape \cite{sem})}. Since a rooted binary tree with $s$ leaves has $s-1$ internal vertices,  $m \leq s-1$.  
Thus, $NT(t)  \geq s!/2^{s-1}$.
\end{proof}

\begin{thm}
Let $T_1$ and $T_2$ be generated from any exchangeable, sampling
consistent distribution on rooted binary trees with $n$ leaves.
For any $\lambda > e \sqrt{2}$ there is a value $m$ such that, for all $n \geq m$,
$$
\mathbb{E}[\mast(T_1,T_2)]  \leq  \lambda  \sqrt{n}.
$$
\end{thm}

\begin{proof}
We  explore the asymptotic behaviour of the quantity
$ \phi_{n,s} = {n \choose s} \frac{2^{s-1}}{s!}$.
Using the inequality $\binom{n}{s} \leq \frac{n^s}{s!}$ and Stirling's approximation, we have:
$$
\phi_{n,s}  \leq  \frac{1}{4 \pi s}  \left(  \frac{2 e^{2}n}{s^{2}} \right)^{s} \theta(s)
$$ 
where $\theta(s)  \sim 1$.  
Hence, $\phi_{n,s}$ tends to zero as an exponential
function of $n$ as $n \rightarrow \infty$.
Since $\phi_{n,s} \geq \psi_{n,s}$ we see that 
$\mathbb{P}[\mast(T_{1}, T_{2}) > \lambda \sqrt{n}]$ tends to zero as an exponential 
function of $n$.  Since $\mast(T_{1}, T_{2}) \leq n$, this implies that
$\mathbb{E}[\mast(T_1,T_2)]  \leq  \lambda  \sqrt{n}$.
\end{proof}


\section*{Acknowledgments}
Work on this project was started at the 2014 NSF-CBMS Workshop
on Phylogenetics at Winthrop University, partially supported by the US National Science Foundation (DMS 1346946).
Lam Ho was partially supported by the NSF Grant DMS 1264153 and the NIH Grant R01 AI107034.
Colby Long was partially supported by the US National 
Science Foundation (DMS 0954865).
Mike Steel was partially supported by the NZ Marsden Fund.  
Katherine St.~John was partially supported by the Simons Foundation.
Seth Sullivant was partially supported by the David and Lucille Packard Foundation and the US National Science Foundation (DMS 0954865). \blue{Finally, we thank the three anonymous reviewers for their helpful comments concerning an earlier version of this manuscript.}

%

\end{document}